\RequirePackage{amsmath}
\documentclass[12pt]{llncs}
\usepackage[colorlinks=true,allcolors=blue,linkcolor=blue,citecolor=blue,bookmarksopen=true]{hyperref}
\usepackage{bookmark}
\usepackage{graphicx}
\usepackage[dvipsnames]{xcolor}
\usepackage{listings}
\usepackage{amssymb}
\usepackage{amsfonts}
\usepackage{mathtools}
\usepackage{maple}
\usepackage[utf8]{inputenc}
\usepackage{breqn}
\usepackage{orcidlink}
\usepackage{mathrsfs}
\usepackage{listings}
\usepackage{microtype}
\usepackage{moreverb}
\usepackage{verbatim}
\usepackage{sverb}
\usepackage{geometry}
\usepackage{color}
\usepackage{makecell}
\usepackage{url}
\usepackage{cleveref}
\usepackage{moreverb}
\usepackage{verbatim}
\usepackage{sverb}
\usepackage{algorithm}
\usepackage{algpseudocode}
\usepackage{algorithmicx}
\usepackage{euscript}
\usepackage{abstract}
\usepackage[T1]{fontenc}
\usepackage{lmodern}
\usepackage{textcomp}
\usepackage{stmaryrd}
\SetSymbolFont{stmry}{bold}{U}{stmry}{m}{n}

\newcommand*{\QEDA}{\hfill\ensuremath{\blacksquare}}

\newcommand*{\CQFD}{\hfill\ensuremath{\square}}

\newcommand{\NN}{\mathbb{N}}

\newcommand{\KK}{\mathbb{K}}
\newcommand{\DD}{\mathbb{D}}

\DeclareMathOperator{\ord}{ord}
\DeclareMathOperator{\res}{res}

\oddsidemargin \evensidemargin
\begin{document}

\title{On Rational Recursion for Holonomic Sequences}

\author{Bertrand Teguia Tabuguia\orcidlink{0000-0001-9199-7077} and James Worrell}

\institute{Department of Computer Science\\
           University of Oxford, UK\\
    \email{\{bertrand.teguia,jbw\}@cs.ox.ac.uk}}

\maketitle              

\begin{abstract}

It was recently conjectured that every component of a discrete-time rational dynamical system is a solution to an algebraic difference equation that is linear in its highest-shift term (a quasi-linear equation). 
We prove that the conjecture holds in the special case of holonomic sequences, which can straightforwardly be 
represented by rational dynamical systems.
We propose two algorithms for converting holonomic recurrence equations into such quasi-linear equations. The two algorithms differ in their efficiency and the minimality of orders in their outputs.

\end{abstract} 

\section{Introduction}

Let us consider a dynamical system of the form
\begin{equation}\label{eq:dynsys1}
\begin{cases}
    s_i(n+1) = R_i(s_1(n),\ldots,s_k(n)),\, i=1,\ldots,k,\\
    a(n) = T(s_1(n),\ldots,s_k(n))
\end{cases},\,\, n\in\NN\coloneqq\{0,1,\ldots\},
\end{equation}
where $T,R_1,\ldots,R_k$ are rational functions over a field $\KK$ of characteristic zero. The vector $\boldsymbol{s}(n)=(s_1(n),\ldots,s_k(n))$ represents the state sequence, and $a(n)$ is called the output sequence. The integer $k$ is the dimension of the system. We are interested in the case where $T$ is the projection of $\boldsymbol{s}(n)$ on one of its components, i.e., $T(\boldsymbol{s}(n))=s_j(n)$ for some $j\in\{1,\ldots,k\}$. The resulting system defines $(a(n))_n$ as a \textit{rational recursive sequence} (or simply ratrec sequence) according to the definition in \cite{clemente2023rational}.  
In control theory, this definition relates to implicit state-space representation \cite[Chapter 3]{hangos2006analysis}. In computer science, sequences of this kind appear as cost-register automata~\cite{AlurDDRY13}
or polynomial automata~\cite{benedikt2017polynomial} over a unary alphabet.
The zeroness problem asks whether such sequence is identically zero\cite{lagarias1988unique,benedikt2017polynomial}. This problem relates to the equivalence problem of classes of automata and grammars \cite{schutzenberger1961definition,worrell2013revisiting}. 
(It can also be contrasted with the Skolem Problem, which asks whether a linear recurrence sequence contains at least one zero term~\cite{ouaknine2015linear}.) 

Observe that from \eqref{eq:dynsys1} we can write $a(n)=T(\boldsymbol{R}^n(\boldsymbol{s}(0)), \boldsymbol{R}=(R_1,\ldots,R_k)$. Hence, another perspective on the zeroness problem arises from the study of the orbit $\{\boldsymbol{R}^n(\boldsymbol{s}(0)),\, n\in\NN\}$. This is particularly investigated in the case of polynomial updates in \eqref{eq:dynsys1} (see \cite{lagarias1988unique,Whang2023}).

It is conjectured in \cite{clemente2023rational} that the output sequence $(a(n))_{n\in\NN}$ (or simply $(a(n))_n$) from \eqref{eq:dynsys1} satisfies a recurrence relation of the form
\begin{equation}
    a(n+l) = r(a(n),\ldots,a(n+l-1)) \label{eq:simprec},
\end{equation}
where $r\in\KK(x_1,\ldots,x_l)$. A sequence satisfying such a recursion is called \textit{simple ratrec}~\cite{clemente2023rational}. Recall that a sequence $(u(n))_n$ is holonomic (or P-recursive) if there exist polynomials $P_0,\ldots,P_l,$ not all zeros such that
\begin{equation}
    P_l(n)\,u(n+l)+\cdots+P_0(n)\,u(n)=0, \forall n\in\NN.
\end{equation}
The maximum degree of the polynomials $P_i,i=0,\ldots,l$ is called the degree of the holonomic equation. It is also the degree of $(u(n))_n$ if $l$ is minimal. Using the change of variables $s_{i+1}(n)=u(n+i), i=0,\ldots,l$, one verifies that the holonomic sequence $(u(n))_n$ may be defined by the system
\begin{equation}\label{eq:holotosys}
\begin{cases}
    s_i(n+1) = s_{i+1}(n),\, i=0,\ldots,l\\
    s_{l+1}(n+1) = - \frac{1}{P_d(s_{l+2}(n))}\left(P_0(s_{l+2}(n))\,s_1(n)+\cdots+P_{l-1}(s_{l+2}(n))\,s_l(n)\right)\\
    s_{l+2}(n+1) = s_{l+2}(n)+1\\
    a(n) = s_{l+1}(n)
\end{cases}.
\end{equation}
Thus, because of the conjecture that every ratrec sequence is simple ratrec, it is natural to investigate the particular case of holonomic sequences.

\section{Problem Statement in Difference Algebra}\label{sec:pbda}

Is every holonomic sequence simple ratrec? We approach this question from the setting of difference algebra \cite{Cohnbook,ovchinnikov2020effective}. We consider the difference ring $(\KK(n),\,\sigma)$, where $\sigma$ denotes the \emph{shift map}, which is defined by $\sigma(f(n))=f(n+1)$ for all $f \in \KK(n)$. (The map $\sigma$ should not be confused with the related shift operator, mainly used in operator algebra.) We look at $s(n+j),j\in\NN$, as a variable in a multivariate polynomial ring. For a given holonomic equation
\begin{equation}\label{eq:holoeq}
    P_l(n)\,s(n+l)+\cdots+P_0(n)\,s(n)=0,
\end{equation}
we consider the ring of difference polynomials $\DD_s\coloneqq \KK(n)[\sigma^{\infty}\left(s(n)\right)]\coloneqq \KK(n)[s(n+\infty)]$, i.e., $\DD_s\coloneqq \KK(n)[s(n),s(n+1),\sigma^2(s(n)),\ldots]$, and the difference ideal $\langle \sigma^{\infty}(p)\rangle \subset \DD_s$, where $p \in \DD_s$ is the left-hand side of \eqref{eq:holoeq}. Note that for any $p\in\DD_s$, the associated equation $p=0$ is equivalent to the one obtained after clearing the denominators. For this reason, we will always assume that our difference polynomials are polynomials in $n$ and some shifts of $s(n)$. A difference polynomial $p\in\DD_s$ is holonomic if the associated equation $p=0$ is a holonomic, i.e. $p$ has degree at most one in each of the indeterminates $s(n),s(n+1),\ldots,s(n+\ell)$.  We similarly say that $p$ is simple ratrec if the equation $p=0$ is simple ratrec. The order $\ord(p)$ of a difference polynomial $p$ is the same as that of the associated difference equation. Its degree $\deg(p)$ is its total degree over $\DD_s$. Unless stated otherwise, for a holonomic $p\in\DD_s$, by degree, we always mean that of its associated equation since, by definition, $\deg(p)=1$. For any $p\in\DD_s$ and a sequence $(a(n))_n$ of general term $a(n)$, $p(a(n))$ denotes the evaluation of $p$ at $(a(n))_n$. If $p(a(n))=0$, we say that $(a(n))_n$ is a zero of $p$. From the analogy to differential algebra, the zeros of $p$ may be called \textit{shift-algebraic functions (or sequences)} (see \cite[Chapter IV]{kolchin1973differential},\cite{levin2008difference,RAB2023}).

A version of the following theorem was proved in \cite[Theorem 7.1]{cadilhac2021polynomial} without explicit use of difference elimination. Our version provides a bound for the order of the equation sought. We adapt a known technique from the differential case with enough details. Further details can be found in \cite[Corollary 3.21]{hong2020global} and \cite[Section 1.7]{Glebnotes}. See also the algorithms from \cite{boulier1995representation,bachler2012algorithmic}.

\begin{theorem}\label{th:cancelpoly} Let $(a(n))_n$ be a ratrec sequence defined by \eqref{eq:dynsys1}. Then there exists a difference polynomial $p\in\DD_s$ of order at most $k$ such that $p(a(n))=0$.
\end{theorem}
\begin{proof} Let $Q$ be the common denominator in \eqref{eq:dynsys1} so that $R_i=r_i/Q$, $T=t/Q$. Define the difference polynomials 
\begin{equation}
    \begin{split}
        &p_i\coloneqq Q\,s_i(n+1)-r_i(\boldsymbol{s}(n)),\,\, i=1,\ldots,k,\\
        &q\coloneqq Q\,a(n)-t(\boldsymbol{s}(n)).
    \end{split}
\end{equation}
The corresponding difference ring is $\KK[a(n+\infty), s_i(n+\infty), i=1,\ldots,k]$. We consider the difference ideal 
\[I\coloneqq \langle \sigma^{\infty}(q),\sigma^{\infty}(p_i),i=1,\ldots,k\rangle \colon \left\{\sigma^{\infty}(Q)\right\}^{\infty},\] 
where ``$\colon \left\{\sigma^{\infty}(Q)\right\}^{\infty}$'' denotes the saturation with $\{\sigma^j(Q),\,j\in\NN\}$. This is the main difference with the differential case, which only requires saturating with the corresponding $\{Q\}$. We take a lexicographic monomial ordering that ranks all shifts of $a(n)$ higher than the other indeterminates. Then, the proof follows from the following elimination:
\[I_k \cap \KK[a(n+\infty)] \neq \{0\}, \]
where $I_k\coloneqq \langle \sigma^{j}(q),\sigma^{j}(p_i),i=1,\ldots,k, j=0,\ldots,k\rangle \colon \left\{\sigma^j(Q),\,j\leq k\right\}^{\infty}$, is the $k$th truncation  of $I$. Indeed, by the first Buchberger's criterion, the generators of $I_k$ thus defined constitute a Gr\"obner basis of $I_k$. Since the polynomials in $\KK[a(n+j),j=0,\ldots,k]$ not reducible with respect to this Gr\"obner basis represent $\KK[s_i(n),\, i=1,\ldots,k]$, the transcendence degree of $\KK[a(n+j), s_i(n+j), i=1,\ldots,k, j=0\ldots k]/I_k$ is exactly $k$. However, the transcendence degree of $\KK[a(n+j),\, j=0,\ldots,k]$ is $k+1$. Therefore, we must have $I_k \cap \KK[a(n+j),\, j=0,\ldots,k] \neq \{0\}$. This fact is well exploited in  \cite{RAB2023,teguia2023arithmetic}. \CQFD
\end{proof}
In general, the equation deduced from \Cref{th:cancelpoly} is not linear in its highest-shift term. Thus, using the ratrec definition of a holonomic sequence may not simplify the problem. With the above formulation, the central question of this paper might be rephrased as follows.

\begin{problem} Let $p\in\DD_s$ be holonomic. Can we find $q\in\langle \sigma^{\infty}(p) \rangle$ such that $q$ is linear in $s(n+\ord(q))$? 
\end{problem}

In the next section, we propose a natural way to attack this problem for lower-degree holonomic difference polynomials.

\section{Lower-Degree Holonomic Difference Polynomials}

Let $p\in\DD_s$ be a holonomic difference polynomial of degree $d\in\NN\setminus \{0\}$, and $l=\ord(p)$. We have 
\begin{equation}
    \begin{split}
        &p=P_l(n)\, s(n+l)+\cdots+P_0(n)\,s(n),\\
        &\sigma^j(p)=P_l(n+j)\, s(n+l+j)+\cdots+P_0(n+j)\,s(n+j).
    \end{split}
\end{equation}
We detail our method for $j\leq 2$. Since $p$ is of degree $d$, we can write $P_i(n)=c_{i,d}\,n^d+c_{i,d-1}\,n^{d-1}+\cdots+c_{i,0}$, $c_{i,j}\in\KK$, $i\leq l, j\leq d$. We view $p$ as a polynomial in $\KK[s(n+l),\ldots,s(n)][n]$, i.e.,
\begin{equation}
    p = \sum_{k=0}^d \left(c_{l,k}\,s(n+l)+\cdots+c_{0,k}\,s(n)\right)\,n^k = \sum_{k=0}^d C_k\cdot S_l\, n^k, \label{eq:redp}
\end{equation}
where $S_{l+k}=(s(n+k),\ldots,s(n+l+k))^T$ and $C_k=(c_{0,k},c_{1,k},\ldots,c_{l,k})$. Using the binomial theorem, one verifies that
\begin{align}
&\sigma(p) = \sum_{k=0}^d\sum_{j=0}^k \binom{d-k+j}{j}\, C_{d-j}\cdot S_{l+1}\, n^{d-k}, \label{eq:sigp}\\
&\sigma^2(p) = \sum_{k=0}^d\sum_{j=0}^k \binom{d-k+j}{j}\, 2^j\, C_{d-j}\cdot S_{l+2}\, n^{d-k}. \label{eq:sig2p}
\end{align}
We thus obtain relations involving many variables. We aim to reduce them to the strict minimum needed to check the existence of a simple ratrec difference polynomial in $\langle \sigma^{\infty}(p)\rangle$. For instance, since $p$ is not simple ratrec, the scalar product $C_k\cdot S_l$ can be seen as a single variable $\alpha_{0,k}, k\leq d$. Our first step is to look for a simple ratrec difference polynomial of order $l+1$. To that end, we make the term $s(n+l+1)$ appears explicitly in $\sigma(p)$ and encapsulate the remaining variables in the new variables $\beta_{l,j}\in\KK$ and $\alpha_{1,j}\in\KK[s(n),\ldots,s(n+l)]$, $j=0,\ldots,d$ such that
\begin{equation}
    \sigma(p) = \sum_{k=0}^d\sum_{j=0}^k \binom{d-k+j}{j}\, \left(\beta_{l,d-j}\,s(n+l+1)+\alpha_{1,d-j}\right)\, n^{d-k} \label{eq:sigpsimp}.
\end{equation}
This reduction of the number of variables leads to the following proposition thanks to elimination techniques with resultants or Gr\"{o}bner bases (see \cite{becker1993grobner}). 
\begin{proposition}\label{prop:deg1} Every holonomic sequence of degree $1$ and order $l\in\NN$ is simple ratrec of order at most $l+1$.
\end{proposition}
\begin{proof}
We use the change of variables preceding \Cref{prop:deg1}. Here $d=1$. We consider the algebraic ideal $\langle p, \sigma(p)\rangle \subset \KK[n,s(n+l+1),\alpha_{0,0},\alpha_{0,1},\alpha_{1,0},\alpha_{1,1}]$ and the lexicographic monomial ordering corresponding to the ranking $n\succ s(n+l+1) \succ\alpha_{0,0} \succ \alpha_{0,1} \succ \alpha_{1,0} \succ \alpha_{1,1}$. The first elimination ideal
\begin{equation}
    \langle p, \sigma(p)\rangle \cap \KK[s(n+l+1),\left(\alpha_{i,j}\right)_{i,j\leq 1}]
\end{equation}
is principal and generated by the non-trivial polynomial given by the resultant of $p$ and $\sigma(p)$: 
\begin{equation}\label{eq:ratrecd1}
\res_n(p,\sigma(p)) \coloneqq \left(-\alpha_{0,0} \beta_{l,1}+\alpha_{0,1} \beta_{l,0}+\alpha_{0,1} \beta_{l,1}\right) s(n+l+1)-\alpha_{0,0} \alpha_{1,1}+\alpha_{0,1} \alpha_{1,0}+\alpha_{0,1} \alpha_{1,1}.
\end{equation}
Observe that $\res_n(p,\sigma(p))$ is linear in $s(n+l+1)$. To conclude, we must verify that the \textit{initial} of \eqref{eq:ratrecd1}, i.e., its leading coefficient when regarded as a polynomial in $s(n+l+1)$, is nonzero or yields another ratrec equation of order less than $l+1$. This is the case since the initial of \eqref{eq:ratrecd1} is linear in the $\alpha_{0,j}$'s, which are themselves linear in the $s(n+j)$'s, $j\leq l$, by definition. Hence, if the initial of \eqref{eq:ratrecd1} vanishes, we obtain a ratrec equation of order $l$. Otherwise, we already have a ratrec equation of order $l+1$.  \CQFD
\end{proof}
The most interesting fact about \Cref{prop:deg1} is that the order of the holonomic difference polynomial is not restrictive.
\begin{example} Let us apply the method of \Cref{prop:deg1} to some basic examples. 
\begin{itemize}
    \item The sequence of Catalan numbers $\frac{1}{n+1}\binom{2\,n}{n}$ is a zero of the difference polynomial
    \begin{equation*}
        (n+2)\,s(n+1)- (4\,n+2)\,s(n).
    \end{equation*}
    We proceed as in the proof of \Cref{prop:deg1} without reducing the number of variables. We consider the ring $\KK[n,s(n+2),s(n+1),s(n)]$ and eliminate $n$ from the ideal. We get:
    \begin{equation}
        s(n+2)=\frac{2\,s(n+1) \left(8\,s(n)+s(n+1)\right)}{10\,s(n)-s(n+1)}.
    \end{equation}
    \item $((-1)^n+n)_{n\in\NN}$ is a zero of the difference polynomial  
    \begin{equation*}
    \left(-2 n-3\right) s\! \left(n\right)-2 s\! \left(n+1\right)+\left(2 n+1\right) s\! \left(n+2\right).
    \end{equation*}
    Our method leads to the constant coefficient (C-finite) relation 
    \begin{equation}
        s(n+3)=s(n+2)+s(n+1)-s(n).
    \end{equation}
     One could recover the same equation using a \textit{desingularization} method \cite{abramov1999desingularization,chen2016desingularization}. Thus, we may detect when a holonomic sequence is C-finite without solving its equation, i.e., recognizing when a holonomic equation has the same solution as a linear recurrence with constant coefficients. This reflects the expansions of some holonomic sequences as linear combinations of exponential polynomials.\QEDA
\end{itemize}
\end{example}
For $(n!^2)_{n\in\NN}$, the difference polynomial obtained with the method of \Cref{prop:deg1} is not simple ratrec. This is expected because $(n!^2)_n$ is a zero of a holonomic difference polynomial of degree $2$. However, we can extend the method for some higher-degree holonomic difference polynomials. Our primary purpose is to detect a general phenomenon that ensures that we can find an elimination ideal containing a ratrec generator. Using Maple \cite{maple}, we could prove this for holonomic difference polynomials of degrees up to $3$.

\begin{proposition}\label{th:deg3th} Every holonomic sequence of degree less or equal to $3$ and order $l\in\NN$ is simple ratrec of order at most $l+2$.
\end{proposition}
\begin{proof} \Cref{prop:deg1} covers the case of degree less or equal to $1$. When the degree $d\in \{2,3\}$, we aim to find a ratrec difference polynomial of order $l+2$. Hence, no reason exists to keep $s(n+l+1)$ explicit anymore. The reduction of the number of variables then consists of rewriting \eqref{eq:redp}, \eqref{eq:sigp}, and \eqref{eq:sig2p} as follows:
\begin{equation}\label{eq:redvard23}
\begin{split}
& p = \sum_{k=0}^d \alpha_{0,k}\, n^k,\,\,\,\, \sigma(p) = \sum_{k=0}^d\sum_{j=0}^k \binom{d-k+j}{j}\, \alpha_{1,d-j}\, n^{d-k},\\
& \sigma^2(p) = \sum_{k=0}^d\sum_{j=0}^k \binom{d-k+j}{j}\, 2^j\, \left(\beta_{l,d-j}\,s(n+l+2)+\alpha_{2,d-j}\right)\, n^{d-k},\\
&\beta_{l,j}\in\KK,\, 
\alpha_{i,j}\in\KK[s(n),\ldots,s(n+l+2)], i\leq 2, j\leq d.
\end{split}
\end{equation}
We consider the algebraic ideal $\langle p, \sigma(p), \sigma^2(p)\rangle \subset \KK[n,s(n+l+2),(\alpha_{i,j})_{i\leq 2,j\leq d}]$ with a lexicographic monomial ordering corresponding to any ranking of the variables such that $n\succ s(n+l+2) \succ \alpha_{i,j}, i\leq 2, j\leq d$. Then, the first elimination ideal 
\begin{equation}
\langle p, \sigma(p), \sigma^2(p)\rangle \cap  \KK[s(n+l+2),(\alpha_{i,j})_{i\leq 2,j\leq 3}], \label{eq:elim2}
\end{equation}
has some non-trivial generators of degree $1$ in $s(n+l+2)$. We give details to conclude the proof when $d=2$. For $d=3$, we leave this part to the reader; supplementary materials to support all the computations of the paper are available at the following GitHub web page \url{https://github.com/T3gu1a/SimpleRatRecAndHolonomic}.

For $d=2$, the elimination ideal \eqref{eq:elim2} has $7$ generators among which $3$ are linear in $s(n+l+2)$. The initial in one of them is given by
\begin{dmath}\label{eq:coeffsl2}
-2 \alpha_{0,0} \alpha_{1,1} \beta_{l,2}+4 \alpha_{0,0} \alpha_{1,2} \beta_{l,1}+\alpha_{0,1} \alpha_{1,0} \beta_{l,2}+\alpha_{0,1} \alpha_{1,1} \beta_{l,2}-4 \alpha_{0,1} \alpha_{1,2} \beta_{l,0}-2 \alpha_{0,1} \alpha_{1,2} \beta_{l,1}-4 \alpha_{0,2} \alpha_{1,0} \beta_{l,1}-\alpha_{0,2} \alpha_{1,0} \beta_{l,2}+8 \alpha_{0,2} \alpha_{1,1} \beta_{l,0}+\alpha_{0,2} \alpha_{1,1} \beta_{l,2}+4 \alpha_{0,2} \alpha_{1,2} \beta_{l,0}-2 \alpha_{0,2} \alpha_{1,2} \beta_{l,1}.
\end{dmath}
The important remark here is that every term $\alpha_{i_1,j_1} \alpha_{i_2,j_2} \alpha_{i_3,j_3}$ or $\alpha_{i_1,j_1} \alpha_{i_2,j_2} \beta_{l,j_3}$ is such that $i_1\neq i_2$, $i_1\neq i_3$, and $i_2\neq i_3$. Therefore, if \eqref{eq:coeffsl2} vanishes, it yields a ratrec equation of order at most $l+1$. Otherwise, we already have a ratrec equation of order $l+2$.
\CQFD
\end{proof}
\begin{example} $(n!^2)_{n\in\NN}$ is a zero of the holonomic difference polynomial $s(n+1)-(n+1)^2\,s(n)$. Proceeding as in the proof of \Cref{th:deg3th} without reducing the number of variables, we get an elimination ideal with $3$ simple ratrec generators. We select the following as our simple ratrec equation satisfied by $n!^2$:
\begin{equation}
s(n+3)=\frac{s(n+2) \left(2\, s(n) \,s(n+1)+2 \,s(n)\,s(n+2)-s(n+1)^{2}\right)}{s(n)\,s(n+1)}.
\end{equation}\QEDA
\end{example}

\section{Proof and Algorithms}

For the sequence $(n!^2+n!)_n$, a zero of a holonomic difference polynomial of degree $4$, we find no ratrec difference polynomial with the method of \Cref{th:deg3th}; hence the need to generalize the strategy for higher degree holonomic sequences. The proof of \Cref{th:deg3th} gives us some insight into this generalization. We expect the elimination ideal to have a simple ratrec generator of order, say, $l+m$, and an initial made of terms of the form
\begin{equation}\label{eq:gencoeff}
    c_I\,\prod_{i\in I} \alpha_{i,j_i},\,\, c_I\in\KK,\, I\subset \{1,\ldots,m\},
\end{equation}
where $\alpha_{i,j_i}$ is a $\KK$-linear combination of $s(n+i),\ldots, s(n+l+i)$, $i\leq m,$ with the appropriate notations as in \eqref{eq:redvard23}. Note that the elimination ideal is computationally hard to obtain for generic holonomic difference polynomials of degrees higher than $3$. One way to arrive at a general proof could come from investigating Buchberger's algorithm in this context \cite{becker1993grobner}. However, this is not the method we use. Our main observation from the previous section is that the obtained simple ratrec difference polynomials are always of order at most $l+d$, where $l$ and $d$ are the respective order and degree of the holonomic difference polynomials considered. We use linear algebra to prove this fact in Theorem \ref{th:holoisratrec}. Our first algorithm is iterative and guarantees finding a simple ratrec equation of minimal order. Given a holonomic difference polynomial $p$, at every iteration $j\in\NN\setminus \{0\}$, the idea is to compute the ideal
\begin{equation}
\langle p, \sigma(p),\ldots, \sigma^j(p)\rangle \cap  \KK[s(n+l+j),s(n+l+j-1),\ldots,s(n)], \label{eq:elimalgo}
\end{equation}
and verify if it contains a simple ratrec generator. If this is the case, we stop and return one such generator (usually the one of minimal degree); otherwise, we increment $j$ and repeat the computation. This algorithm is implemented by the new command \texttt{HoloToSimpleRatrec} with option \texttt{method=GB} (standing for Gr\"obner bases) from the package \texttt{NLDE} \cite{teguia2023operations}. We refer to this algorithm or its implementation as the \textit{Gr\"obner bases method}. As inputs, it takes the holonomic difference polynomial (equation) and the indeterminate term (like $s(n)$). We use the degree of the input equation as a bound for the number of iterations. The user may specify a lower bound with the optional argument \texttt{userbound=m}, where \texttt{m} is the proposed value. \Cref{ex:algoexpl} gives some conversions obtained with this method. We keep the examples simple to save some space. For the sequence $(n!^2+n!)_n$, the algorithm returns a $6$th-order simple ratrec equation with large integer coefficients. The discussion in \cite[p. 90]{neun1990very} may explain the presence of such coefficients.

\begin{example}[Automatic conversion with Gr\"obner bases]\label{ex:algoexpl}\item 
\begin{itemize}
    \item Case of $n!^3$:

\begin{lstlisting}
> with(NLDE, HoloToSimpleRatrec): #loading the implementation
> p1 := s(n+1) - (n+1)^3*s(n): #the holonomic difference polynomial
> HoloToSimpleRatrec(p1,s(n),method=GB) #use of our implementation
\end{lstlisting}
\begin{small}
\begin{equation}\label{eq:maple1}
s\! \left(n+3\right)=\frac{s\! \left(n+2\right) \left(4 s\! \left(n\right) s\! \left(n+1\right)^{2}-4 s\! \left(n\right) s\! \left(n+2\right)^{2}+s\! \left(n+1\right)^{3}+s\! \left(n+1\right)^{2} s\! \left(n+2\right)\right)}{s\! \left(n+1\right) \left(s\! \left(n\right) s\! \left(n+1\right)-s\! \left(n\right) s\! \left(n+2\right)-2 s\! \left(n+1\right)^{2}\right)}
\end{equation}    
\end{small}

\item $\left(n^{2}+\sin\left(n\pi/4\right)^{2}\right)_{n\in\NN}$ is a zero of the $3$rd-order difference polynomial assigned to \texttt{p2} below. Our implementation finds a C-finite recurrence of order $5$ from it.
\begin{lstlisting}
> p2 := (-2*n^2 - 8*n - 11)*s(n) + (2*n^2 + 4*n + 5)*s(n + 1) 
+ (-2*n^2 - 8*n - 11)*s(n + 2) + (2*n^2 + 4*n + 5)*s(n + 3):
> HoloToSimpleRatrec(p2,s(n),method=GB)
\end{lstlisting}
\begin{equation}\label{eq:maple3}
s\! \left(n+5\right)=s\! \left(n\right)-3 s\! \left(n+1\right)+4 s\! \left(n+2\right)-4 s\! \left(n+3\right)+3 s\! \left(n+4\right)
\end{equation}
\end{itemize}
\QEDA
\end{example}
Due to the use of Gr\"obner bases, our implementation \texttt{HoloToSimpleRatrec/method=GB} tends to be slow for holonomic difference polynomials of degrees greater than $4$. If we neglect the minimality of the order for the simple ratrec equation sought, then a more efficient algorithm can be deduced from the proof of \Cref{th:holoisratrec}.

\begin{theorem}\label{th:holoisratrec} Every holonomic sequence of order $l$ and degree $d$ is simple ratec of order at most $l+d$.   
\end{theorem}
\begin{proof} Let $p\in\DD_s$ be holonomic of degree $d$ and order $l$. We see $p$ and its shifts as polynomials in $n$ such that
\begin{equation}\label{eq:shiftgamma}
    \sigma^j(p) = \sum_{k=0}^d \gamma_{j,k}\,n^k,\,\, j=0,\ldots,d,
\end{equation}
where $\gamma_{j,k}\in\KK[s(n+j),\ldots,s(n+l+j)]$. We consider the associated equations for $j<d$ and write them as follows
\begin{equation}\label{eq:nksys}
   \sum_{k=1}^d \gamma_{j,k}\,n^k = -\gamma_{j,0},\,\, j=0,\ldots,d-1.
\end{equation}
This is a linear system of $d$ equations in the unknown $(n,n^2,\ldots,n^d)^T$. We claim that the matrix of the system has full rank due to the shifts involved in the $\gamma_{j,k}$'s. Indeed, each $\gamma_{j,k}$ has the following form:
\begin{equation}\label{eq:gammacsdef}
    \gamma_{j,k} \coloneqq \sum_{i=0}^{l} c_{i,k}\,s(n+i+j) = C_k\cdot S_{l+j},
\end{equation}
where $c_{i,k}$ is the constant coefficient of $n^k$ in the polynomial coefficient of $s(n+i)$ in $p$, $i=0,..,l$. We use the same vector notations from $\eqref{eq:redp}$ for $C_k$ and $S_{l+j}$.

Let $\Gamma_1,\ldots,\Gamma_d$ be the rows of the matrix from \eqref{eq:nksys}. Observe that 
\begin{equation}\label{eq:Gammacsdef}
    \Gamma_j = (C_1\cdot S_{l+j}, C_2\cdot S_{l+j},\ldots, C_d\cdot S_{l+j}).
\end{equation}
Suppose there exist $\lambda_1,\lambda_2,\ldots,\lambda_d\in \KK(s(n),\ldots,s(n+d+l-1))$ such that $\sum_{j=1}^d \lambda_j\,\Gamma_j=0$. We look at the components of $v\coloneqq \sum_{j=1}^d \lambda_j\,\Gamma_j = (v_1,v_2,\ldots,v_d)$ and derive conditions on the $\lambda$'s to obtain the zero vector. We consider a component of $v$ in which $s(n)$ and $s(n+l+d-1)$ (the lowest and the largest shifts of $s(n)$) occur in the resulting linear combination. Such a component necessarily exists because $p\neq 0$. Without loss of generality, let us assume that $v_1$ is such a component. Then we have that 
\begin{equation}\label{eq:l1ldeq}
    v_1 - \left(\lambda_1\,c_{1,1}\,s(n) + \lambda_d\,c_{l,1}\,s(n+l+d-1)\right)
\end{equation}
is free of $s(n)$ and $s(n+l+d-1)$. Thus to have $v_1=0$, we must have $\lambda_1=\lambda_d=0$. We repeat this reasoning considering $\lambda_1=\lambda_d=0$. If $s(n+1)$ and $s(n+l+d-2)$ are the respective lowest and largest shifts of $s(n)$ in $v_1$, then
\begin{equation}\label{eq:l2ld1eq}
    v_1 - \left(\lambda_2\,c_{1,1}\,s(n+1) + \lambda_{d-1}\,c_{l,1}\,s(n+l+d-2)\right)
\end{equation}
is free of $s(n+1)$ and $s(n+l+d-2)$, and we deduce that $\lambda_2=\lambda_{d-1}=0$. Therefore by the same process it follows that $\lambda_1=\lambda_2=\cdots=\lambda_d=0$. 

Hence for all positive integers $k\leq d$, $n^k\in\KK(s(n),\ldots,s(n+d+l-1))$. Finally, we plug these rational functions into $\sigma^d(p)$. Since all the $n^k$'s are free of $s(n+l+d)$, the resulting difference polynomial yields a simple ratrec equation of order $l+d$.\CQFD
\end{proof}
The above proof is constructive and yields an algorithm to convert any holonomic equation into a simple ratrec equation. We mention that the analog of \Cref{th:holoisratrec} in the differential case is relatively simple and our result cannot be seen as an adaptation of the differential context. Indeed, the first derivative of any algebraic differential equation is linear in its highest-order term. Hence, in that context, the problem is reduced to finding an algebraic differential equation from a holonomic differential equation. Algorithms for such computations are described in \cite{jimenez2020some,teguia2023operations}.

The default value of \texttt{method} for our command \texttt{HoloToSimpleRatrec} is $\texttt{LA}$, which stands for linear algebra. We refer to this method as the \textit{linear algebra method}. Let us present some computations.
\begin{example}[Automatic conversion with linear algebra]\label{ex:exla} For the search for a holonomic equation satisfied by a given general term, we use the algorithm from \cite{teguia2024computing}.
\begin{itemize}
\item $\left(\binom{2\,n}{n}\,\binom{3\,n}{n}\right)_{n\in\NN}$ is a zero of the difference polynomial assigned to \texttt{p3} below. Both of our methods return the same $3$rd-order simple ratrec equation.
\begin{lstlisting}
> with(NLDE, HoloToSimpleRatrec): 
> p3 := s(n + 1)*(n + 1)^2 - 3*(3*n + 1)*(3*n + 2)*s(n):
> HoloToSimpleRatrec(RE3,s(n),method=LA)
\end{lstlisting}
\begin{small}
\begin{equation}
    s\! \left(n+3\right)=\frac{3 s\! \left(n+2\right) \left(26244 s\! \left(n\right) s\! \left(n+1\right)-702 s\! \left(n\right) s\! \left(n+2\right)-378 s\! \left(n+1\right)^{2}+13 s\! \left(n+1\right) s\! \left(n+2\right)\right)}{5508 s\! \left(n\right) s\! \left(n+1\right)-201 s\! \left(n\right) s\! \left(n+2\right)-84 s\! \left(n+1\right)^{2}+4 s\! \left(n+1\right) s\! \left(n+2\right)}
\end{equation}
\end{small}
\item The sequence $\left(n^4/2^n+3^n\right)_n$ is a zero of the second-order holonomic difference polynomial assigned to \texttt{p4} below. The linear algebra method finds a C-finite recurrence equation.
\begin{lstlisting}
> p4 := (15*n^4+48*n^3+36*n^2-24*n-30)*s(n) - (7*n^2+4*n+4)
*(5*n^2-4*n-4)*s(n + 1) + (10*n^4-8*n^3-12*n^2-8*n-2)*s(n + 2):
> HoloToSimpleRatrec(RE4,s(n),method=LA)
\end{lstlisting}
\begin{small}
\begin{equation}\label{eq:maple4}
s\! \left(n+6\right)=-10 s\! \left(n+4\right)-\frac{3 s\! \left(n\right)}{32}+\frac{31 s\! \left(n+1\right)}{32}-\frac{65 s\! \left(n+2\right)}{16}+\frac{35 s\! \left(n+3\right)}{4}+\frac{11 s\! \left(n+5\right)}{2}
\end{equation}  
\end{small}

The Gr\"obner bases method finds a nonlinear $4$th-order equation for this example.
\item Inhomogeneous equations are also allowed as inputs. This may be used to find a C-finite equation satisfied by a given polynomial sequence. Indeed, it is easy to verify that the linear algebra method always finds a C-finite recurrence for any input of the form $s(n)+P$ (or $s(n)+P=0$ or $s(n)=P$), where $P\in\KK[n]$. This is well-known: every polynomial sequence satisfies a C-finite recurrence equation. For the following example, we generate a random polynomial of degree $5$ and find its corresponding C-finite recurrence of order $5$.
\begin{lstlisting}
> p5:=s(n)+randpoly(n,degree=5)
\end{lstlisting}
\begin{equation}
    \texttt{p5}\coloneqq s\! \left(n\right)-96 n^{5}+9 n^{4}+93 n^{3}+81 n^{2}+16 n+86
\end{equation}

\begin{lstlisting}
> HoloToSimpleRatrec(p5,s(n),method=LA)
\end{lstlisting}
\begin{equation}
s\! \left(n+5\right)=11520+5 s\! \left(n+4\right)-10 s\! \left(n+3\right)+s\! \left(n\right)-5 s\! \left(n+1\right)+10 s\! \left(n+2\right)
\end{equation}
\end{itemize}
\QEDA
\end{example}

\section{Concluding Remarks}

We described and presented our implementation of two algorithms for computing simple ratrec equations from holonomic recurrence equations. One is based on elimination with Gr\"obner bases, and the other on linear algebra. We now compare the timings of both methods. We generate $5$ random holonomic recurrence equations of order at most $10$ and degree at most $5$ that serve as inputs. We do this several times and collect the timings that average the performance of the two methods. The timings of \Cref{tab:comp} are obtained from the following holonomic equations. A Maple worksheet supporting these computations can be found in our GitHub repository at \url{https://github.com/T3gu1a/SimpleRatRecAndHolonomic}.
\begin{enumerate}
    \item 
    \begin{dmath}
        -n^{2} s\! \left(n\right)+\left(n^{2}-1\right) s\! \left(n+1\right)+\left(n^{2}-n-1\right) s\! \left(n+2\right)+n s\! \left(n+3\right)+\left(n+1\right) s\! \left(n+4\right)+\left(-n^{2}-1\right) s\! \left(n+5\right)-n^{2} s\! \left(n+6\right)+\left(n^{2}+n\right) s\! \left(n+7\right)+\left(-1-n\right) s\! \left(n+8\right)+\left(n^{2}-n\right) s\! \left(n+9\right)+\left(-n^{2}+1\right) s\! \left(n+10\right)=0 \label{eq:reeq1}
    \end{dmath}
    \item 
    \begin{dmath}
        \left(-n^{2}+n+1\right) s\! \left(n\right)+\left(n^{3}-n^{2}+n+1\right) s\! \left(n+1\right)+\left(n^{2}-n-1\right) s\! \left(n+2\right)+\left(n^{3}-n^{2}-n\right) s\! \left(n+3\right)+n s\! \left(n+4\right)+\left(n^{3}-n-1\right) s\! \left(n+5\right)+\left(-n^{3}+n^{2}-n+1\right) s\! \left(n+6\right)+\left(n^{2}-n\right) s\! \left(n+7\right)+\left(-n^{3}+n^{2}-n+1\right) s\! \left(n+8\right)+\left(-n^{3}-n^{2}\right) s\! \left(n+9\right)=0\label{eq:reeq2}
    \end{dmath}
    \item 
    \begin{dmath}
        \left(n^{4}-n^{3}-n^{2}-1\right) s\! \left(n\right)+\left(n^{4}-n^{2}-n-1\right) s\! \left(n+1\right)+\left(-n^{4}+n^{3}-n^{2}+n\right) s\! \left(n+2\right)+\left(n^{2}+1\right) s\! \left(n+3\right)+\left(-n^{4}-n^{3}+n^{2}+n+1\right) s\! \left(n+4\right)+\left(-n^{4}+n^{2}+n+1\right) s\! \left(n+5\right)+\left(-n^{4}+n^{3}+n^{2}\right) s\! \left(n+6\right)+\left(-n^{4}+n^{2}+n\right) s\! \left(n+7\right)+\left(n^{3}+n^{2}-n\right) s\! \left(n+8\right)=0 \label{eq:reeq3}
    \end{dmath}
    \item 
    \begin{dmath}
        \left(-n^{4}+n^{2}-n+1\right) s\! \left(n\right)+\left(n^{4}-n^{3}+n^{2}+n\right) s\! \left(n+1\right)+\left(-n^{5}-n^{4}+n^{2}+n\right) s\! \left(n+2\right)+\left(n^{4}+n^{2}-1\right) s\! \left(n+3\right)+\left(-n^{4}-n^{2}-n+1\right) s\! \left(n+4\right)+\left(-n^{4}-n^{3}+n^{2}-n-1\right) s\! \left(n+5\right)+\left(-n^{3}+n^{2}+n-1\right) s\! \left(n+6\right)=0 \label{eq:reeq4}
    \end{dmath}
    \item 
    \begin{dmath}
        \left(-1-n\right) s\! \left(n\right)-n s\! \left(n+1\right)+\left(-n+1\right) s\! \left(n+2\right)+\left(-1-n\right) s\! \left(n+3\right)+\left(n+1\right) s\! \left(n+4\right)-n s\! \left(n+5\right)+s\! \left(n+6\right)+\left(-1-n\right) s\! \left(n+7\right)=0 \label{eq:reeq5}
    \end{dmath}
\end{enumerate}

\begin{table}[ht!]
\centering
\caption{Outputs ($t, \, \ord,\, \deg\,$) of \texttt{method=GB} and \texttt{method=LA} for the differential polynomials in \eqref{eq:reeq1}--\eqref{eq:reeq5}. We display the CPU time~$t$ in seconds, the order~$\ord$, and the degree~$\deg$ of the computed simple ratrec equation. We stopped the computation if we did not obtain a result after $300$ seconds (5 minutes). This is indicated as ``$300+$'' in the table.}  
\label{tab:comp}
\begin{tabular}{|c|c|c|c|c|}
\hline
\multicolumn{1}{|l|}{Equation} & Order & Degree & \begin{tabular}[c]{@{}c@{}}Output\\ \texttt{method=GB}\end{tabular} & \begin{tabular}[c]{@{}c@{}}Output\\ \texttt{method=LA}\end{tabular} \\ \hline
$\eqref{eq:reeq1}$             & $10$  & $2$    & $(0.218,12,3)$                        & $(0.031,12,3)$                                            \\ \hline
$\eqref{eq:reeq2}$             & $9$   & $3$    & $(27.547,11,5)$                       & $(0.188,12,4)$                                             \\ \hline
$\eqref{eq:reeq3}$             & $8$   & $4$    & $(300+,-,-)$                          & $(0.734,12,5)$                                             \\ \hline
$\eqref{eq:reeq4}$             & $6$   & $5$    & $(300+,-,-)$                          & $(2.203,11,6)$                                             \\ \hline
$\eqref{eq:reeq5}$             & $7$   & $1$    & $(0.,8,2)$                            & $(0.,8,2)$                                                 \\ \hline
\end{tabular}
\end{table}

As illustrated by \Cref{tab:comp}, while the linear algebra method has better efficiency, the Gr\"obner bases method always returns an equation of minimal order. Indeed, although \Cref{prop:deg1} is a special case of \Cref{th:holoisratrec}, \Cref{th:deg3th} gives a better bound for holonomic difference equations of degree $3$. This shows how tight the bound given by the Gr\"obner bases method can be. The difference between the orders of their outputs may be greater than $1$ in absolute value. For instance, for $(n!^4)_{n\in\NN}$, \texttt{HoloToSimpleRatrec/method=GB} finds a $3$rd-order simple ratrec equation while \texttt{HoloToSimpleRatrec/method=LA} computes a $5$th-order equation.
 
As a final remark, we point out that our algorithms may be used to automatically generate ``\textit{Somos-like}'' sequences \cite{malouf1992integer} like in \cite{ekhad2014generate}. The idea is to take any holonomic equation of the form $a(n+l)+P_{n+l-1}(n)\,a(n+l-1)+\cdots+P_0(n)\,a(n)=0$ with only integral initial values. Then, the resulting sequence is integral, and one can use our algorithms to find the desired rational recursion, thus defining an integer sequence with a rational recursion.

\bigskip

\section*{Acknowledgment.} We thank Gleb Pogudin for a helpful conversation about \Cref{th:cancelpoly} and Somos-like sequences. This work was supported by UKRI Frontier Research Grant EP/X033813/1.

\end{document}